\documentclass[submission,copyright,creativecommons,sharealike]{eptcs}
\providecommand{\event}{LFMTP 2021} 

\usepackage{listings}
\usepackage{stackengine,scalerel}
\newcommand\textaltcolon{\ensurestackMath{\stackon[1.5ex]{\circ}{\circ}}}
\newcommand\altcolon{\savestack\Tmp{\raisebox{-.7pt}{$\textaltcolon$}}%
	\dp\Tmpcontent=\dimexpr\dp\Tmpcontent-.7pt\relax%
	\mathrel{\scalerel*{\Tmp}{:}}}

\usepackage[hide]{ed}
\usepackage{paralist}
\usepackage{url}
\usepackage{wrapfig}
\usepackage{amsmath,amssymb,amsthm}

\theoremstyle{plain}
\newtheorem{theorem}{Theorem}
\newtheorem{lemma}[theorem]{Lemma}

\theoremstyle{definition}
\newtheorem{definition}{Definition}

\theoremstyle{remark}

\newtheorem{example}{Example}

\usepackage{multicol}

\usepackage{xcolor}
\hypersetup{linkcolor=red,citecolor=blue,urlcolor=gray,colorlinks,breaklinks,bookmarksopen,bookmarksnumbered}

\usepackage{basics}
\usepackage{mmtlf}

\usepackage{local}

\title{Systematic Translation of Formalizations of Type Theory from Intrinsic to Extrinsic Style}
\def\titlerunning{Translating Intrinsic to Extrinsic Type Theories}

\author{Florian Rabe \qquad\qquad Navid Roux
	\institute{University Erlangen-Nuremberg}
	\institute{Department of Computer Science\\
		University Erlangen-Nuremberg\\
		Erlangen, Germany}
	\email{florian.rabe@fau.de \qquad\qquad navid.roux@fau.de}
}

\def\authorrunning{F. Rabe and N. Roux}

\begin{document}
\let\numberedparagraph=\paragraph
\renewcommand{\paragraph}[1]{\vspace{.3cm}\noindent\textbf{#1}\;}

\maketitle
\begin{abstract}
	Type theories can be formalized using the intrinsically (hard) or the extrinsically (soft) typed style. In large libraries of type theoretical features, often both styles are present, which can lead to code duplication and integration issues.

	We define an operator that systematically translates a hard-typed into the corresponding soft-typed formulation.
	Even though this translation is known in principle, a number of subtleties make it more difficult than naively expected.
	Importantly, our translation preserves modularity, i.e., it maps structured sets of hard-typed features to correspondingly structured soft-typed ones.

	We implement our operator in the \mmt system and apply it to a library of type-theoretical features.
\end{abstract}

\section{Introduction}\label{sec:intro}
\paragraph{Motivation and Related Work}
Soft type theory goes back to Curry's work \cite{curry}, where typing is a meta-language (ML) predicate $\of:\term\to\tp\to\prop$ between object-language (OL) terms and types.
This is also called extrinsic typing.
This leads to a deep embedding of typing where the OL judgment $t:A$ corresponds to the existence of a typing proof, i.e., an ML term witnessing $\of t\,A$.

Hard type theory goes back to Church's work \cite{churchtypes}, where typing is a function from terms to types.
If that function is sufficiently simple, formalizations in a dependently-typed ML like LF \cite{lf} may be able to capture it directly in the framework's type system.
Concretely, such a representation uses a shallow embedding of typing where an object-language typing judgment $t:A$ corresponds to the meta-language (ML) typing judgment $t:\tm A$.
Only well-formed terms can be encoded at, and the OL type of a term can be read off of its ML type.
This is also called intrinsic typing.
If ML type-checking is decidable, that means OL must be as well.

Soft typing is more expressive and flexible than hard typing.
But
\begin{compactitem}
	\item introducing an OL variable $x$ of type $a$ requires two ML variables of $x:\term, x^*:\of x\, a$,
	\item correspondingly, substituting a term for a variable requires the term and a proof of its typing judgment,
	\item type-checking is reduced to ML non-emptiness-checking, which is usually undecidable.
\end{compactitem}

In the LATIN project \cite{CHKMR:latinabs:11} going back to ideas developed in the Logosphere project \cite{logosphere}, we built a highly modular library of formalizations of logics and type theories.
The goal was to create a library of Little Logics (in the style of Little Theories \cite{littletheories}), each formalizing one language feature such as product types, which can be combined to form concrete systems.
This allows the reuse, translation, and combination of formalizations in the style of \cite{lfencodings,devgraphshiding}.
It helps the meta-theoretical analysis as each modular construction is itself a meta-theorem, e.g., reusing the formalization of a language feature implies that two languages share that feature, and translations between languages can allow moving theorems across formal systems  \cite{hol_nuprl,lfcut}.

Due to the incomparable advantages of soft and hard encodings, we had to formalize each feature in both styles.
This led not only to a duplication of code but also caused significant maintenance problems.
In particular, it is difficult to ensure coherent encoding styles (e.g., naming conventions, order of arguments, notations, etc.) in such a way that the two sets of encodings are related systematically.

For every pair of a hard- and a soft-typed encoding of the same language feature, there is a type erasure translation from the former to the latter with an associated type preservation property.
Such translations have been investigated in various forms, see e.g., \cite{paramreal_bernardy} for a systematic study in the form of realizability theories.
We will cast the type preservation as a logical relation proof as formalized in an LF-based logical framework in \cite{RS:logrels:12}.
Notably, given the hard-typed encoding, it is possible to derive the soft-typed one, the erasure translation, and the preservation proof automatically.
We call this derivation \emph{softening}.
Systematic softening not only greatly reduces the encoding effort but simplifies maintenance and produces more elegant code.

%

\paragraph{Contribution and Overview}
We define an operator \Soften in the logical framework LF~\cite{lf}.
Despite being conceptually straightforward, softening is a rather complex process, and an ad-hoc implementation, while possible, would be error-prone and hard to maintain.
Therefore, we employ a systematic approach for deriving the softening operator that constructs the logical relation proof along with the softened theory.
A particular subtlety was to ensure the generated code to still be human-readable. That required softening to consider pragmatic aspects like notations and choice of implicit arguments.

Our work is carried out under the LATIN2 header, which aims at a complete reimplementation of the LATIN library.
While LATIN worked with modular Twelf~\cite{RS:twelfmod:09}, LATIN2 uses the \mmt/LF incarnation of LF~\cite{rabe:howto:14}.
In addition to an implementation of LF and a module system, \mmt provides a framework for diagram operators \cite{RR:linearops:20}, which supports the meta-theory and implementation of operators that systematically derive formalizations from one another.
It also makes it easy to annotate declarations, which we will use to guide the softening operator in a few places.

Importantly, these diagram operators are functorial in the category of LF theories and theory morphisms.
That enables scaling them up to entire libraries in a way that preserves modularity.
That is important to derive human-readable formalizations.

Sect.~\ref{sec:mmt} introduces \mmt/LF.
Sect.~\ref{sec:logrel} shows the key definition of the softening operator, and Sect.~\ref{sec:diagrams} establishes meta-theoretical properties that allow lifting it to libraries.
Sect.~\ref{sec:impl} shortly sketches our implementation in the \mmt system.

\section{The MMT Framework and Basic Formalizations}\label{sec:mmt}
\mmt \cite{rabe:howto:14} is a framework for designing and implementing logical frameworks.
To simplify, we only use the implementation of LF that comes with \mmt's standard library, and restrict the grammar to the main features of \mmt/LF:
We assume the reader is familiar with LF (see e.g., \cite{lf}) and only recap the notions of theories and morphisms that \mmt adds on top.

\begin{commgrammar}
\gprod{\Delta}{\cdot}{diagrams}\\
\galtprod{\Delta,\;\ithy{T}{\Theta}}{theory definition}\\
\galtprod{\Delta,\;\imor{m}{S}{T}{\theta}}{morphism definition}\\
\gprod{\Theta}{\cdot}{declarations in a theory}\\
\galtprod{\Theta,\;c:A[=t]}{typed, optionally defined constants}\\
\galtprod{\Theta,\;\iincl{S}{}}{include of a theory}\\
\gprod{\theta}{\cdot \bnfalt \theta,\;c=t\bnfalt \theta,\iincl{m}{}}{declarations in a morphism}\\
\gprod{\Gamma}{\cdot \bnfalt \Gamma,x:A}{contexts}\\
\gprod{t,A,f}{c\bnfalt x \bnfalt \type\bnfalt \kind \bnfalt \tlam[A]{x}t \bnfalt \tPi[A]{x}B \bnfalt f\,t}{LF expressions}\\
\end{commgrammar}

\paragraph{Theories}
An \mmt/LF theory is ultimately a list of \textbf{constant} declarations $c:A[=t]$ where the definiens $t$ is optional.
A constant declaration may refer to any previously declared constant.
LF provides the primitives of a dependently typed $\lambda$-calculus, namely universes $\type$ and $\kind$, function types $\tPi[A]{x}B$, abstraction $\tlam[A]{x}t$ and application $f\,t$.
In a constant declaration $c:A$, we must have $A:\type$ or $A:\kind$, and in a variable binding $x:A$, we must have $A:\type$.
As usual, \mmt/LF allows writing $A\to B$ for $\tPi[A]{x}B$ and omitting inferable brackets, arguments, and types.
If we need to be precise about \textbf{typing}, we write $\oftype{T}{\Gamma}{t}{A}$ for the typing judgment between two expressions that may use all constants from theory $T$ and all variables from context $\Gamma$.

A theory $T$ may \textbf{include} a previously defined theory $S$,
which makes all constants of $S$ available in $T$ as if they were declared in $T$.

\begin{example}\label{ex:thys}
	We give theories formalizing hard- and soft-typed type theories.
	The left shows the common theory $\Proofs$ that formalizes proofs in standard LF fashion using the judgments-as-types principle:
	$\ded P$ is the type of proofs of the proposition $P:\prop$, i.e., $\ded\,P$ is non-empty iff $p$ is provable.
	\HTyped formalizes hard typing, also called intrinsic or Church typing, where typing is a function from terms to types, i.e., every term has a unique type that can be inferred from it.
	That enables the representation of object language terms $t:a$ as LF terms $t:\tm a$.
	And \STyped formalizes soft typing, also called extrinsic or Curry typing, where typing is a relation between terms and types, i.e., a term may have multiple or no types.
	That corresponds to a representation of an object language term $t:a$ in LF as an untyped term $t:\term$ for which a proof of $\ded\of t\,a$ exists.

\begin{figure}[h]
\vspace{-1em}
\begin{multicols}{3}
\begin{mmtmods}
\mthy{\Proofs}{}\\
\mcons{\prop}{\type}\\
\mcons{\dedN}{\prop\to\type}\\
\mend\\{}\\
\end{mmtmods}
\begin{mmtmods}
\mthy{\HTyped}{}\\
\mincl{\Proofs}{}\\
\mcons{\tp}{\type}\\
\mcons{\tm}{\tp\to\type}\\
\mend\\
\end{mmtmods}
\begin{mmtmods}
\mthy{\STyped}{}\\
\mincl{\Proofs}{}\\
\mcons{\tp}{\type}\\
\mcons{\term}{\type}\\
\mcons{\ofN}{\term\to\tp\to\prop}\\
\mend
\end{mmtmods}
\end{multicols}
\end{figure}
\end{example}

\paragraph{Morphisms}
A morphism $m:S\to T$ represents a compositional translation of all $S$-syntax to $T$-syntax.
We spell out the definition and key property:

\begin{definition}\label{def:mor}
A morphism $m:S\to T$ is a mapping of $S$-constants to $T$-expressions such that for all $S$-constants $c:A$ we have $\oftype{T}{}{m(c)}{\ov{m}(A)}$
where $\ov{m}$ maps $S$-syntax to $T$-syntax as defined in Fig.~\ref{fig:mor}.
In the sequel, we write $m$ for $\ov{m}$.
\newcommand{\m}{\ov{m}}

\begin{figure}[hbt]
\begin{multicols}{2}
\noindent\[\eqns{
\txt{constants of $S$}\\
\m(c) & m(c)\\[.2cm]
\txt{other expressions}\\
\m(x) & x\\
\m(\type) & \type\\
\m(\tPi[A]{x}B) & \tPi[\m(A)]{x}\m(B)\\
\m(\tlam[A]{x}t)& \tlam[\m(A)]{x}\m(t)\\
\m(f\,t) & \m(f)\,\m(t)\\[.2cm]
\txt{contexts}\\
\m(\cdot) &\cdot\\
\m(\Gamma,x:A) &  \m(\Gamma), x:\m(A)
}\]
\[\eqns{
\txt{theories that include $S$}\\
\m(E=\{\ldots,D_i,\ldots\})&E^m=\{\ldots,\m(D_i),\ldots\}\\
\m(\incl{S})&\incl{T}\\ 
\m(c:A[=t])&c:\m(A)[=\m(t)]\\
\m(\incl{E})&\incl{E^m}\\[.2cm]
\txt{constants of a theory including $S$}\\
\m(c)&c\\
\txt{}
}\]
where $E^m$ generates a fresh name for the translated theory
\end{multicols}
\caption{Map induced by a Morphism}\label{fig:mor}
\end{figure}
\end{definition}

\begin{theorem}\label{thm:mor}
For a morphism $m:S\to T$ and a theory $E$ that includes $S$, if $\oftype{E}{\Gamma}{t}{A}$, then $\oftype{E^m}{m(\Gamma)}{m(t)}{m(A)}$.
In particular for $E = S$, we have $\oftype{T}{m(\Gamma)}{m(t)}{m(A)}$.
\end{theorem}

In terms of category theory, a morphism $m$ induces a \textbf{pushout functor} $\PO{m}$ from the category of theories including $S$ to the category of theories including $T$.
As a functor, $m$ extends to diagrams, i.e., any diagram of theories $E$ including $S$ and morphisms between them is mapped to a corresponding diagram of theories $E^m$ including $T$.
Moreover, for each $E$, $m$ extends to a morphism $E\to E^m$ that maps every $S$-constant according to $m$ and every other constant to itself.
Each of these morphisms maps $E$-contexts/expressions to $E^m$ and that mapping preserves all judgments.
These morphisms form a natural transformation, and we speak of a \textbf{natural functor}.

\begin{figure}[hbt]
\begin{multicols}{2}
\begin{mmtmods*}
\mmor{\TE}{\HTyped}{\STyped}\\
\mincl{\Proofs}{}\\
\mconsd{\tp}{}{\tp}\\
\mconsd{\tm}{}{\tlam[\tp]{a}\term}\\
\mend\\
\end{mmtmods*}{}\\{}\\
\begin{mmtmods*}
	\mthy{\HProd}{}\\
	\mincl{\HTyped}{}\\
	\mcons{\prod}{\tp\to\tp\to\tp}\\
	\mcons{\pair}{\tPi{a,b}\tm a\to\tm b\to\tm \prod\,a\,b}\\
	\mcons{\projL}{\tPi{a,b}\tm \prod\,a\,b\to\tm a}\\
	\mcons{\projR}{\tPi{a,b}\tm \prod\,a\,b\to\tm b}\\
	\mend\\
\end{mmtmods*}
\begin{mmtmods*}
\mthy{\HProd^\TE}{}\\
\mincl{\STyped}{}\\
\mcons{\prod}{\tp\to\tp\to\tp}\\
\mcons{\pair}{\tPi{a,b}\term\to\term\to\term}\\
\mcons{\projL}{\tPi{a,b}\term\to\term}\\
\mcons{\projR}{\tPi{a,b}\term\to\term}\\
\mend
\end{mmtmods*}
\begin{mmtmods*}
\mmor{\TE_\HProd}{\HProd}{\HProd^\TE}\\
\mincl{\TE}{}\\
\mconsd{\prod}{}{\prod}\\
\mconsd{\pair}{}{\pair}\\
\mconsd{\projL}{}{\projL}\\
\mconsd{\projR}{}{\projR}\\
\mend\\
\end{mmtmods*}
\end{multicols}
\caption{\label{fig:te-pushout}Pushout along the Type Erasure Morphism}
\end{figure}

\begin{example}[related to Fig.~\ref{fig:te-pushout}]\label{ex:te}
The type erasure translation $\TE:\HTyped\to\STyped$ maps types $a:\tp$ to types $\TE(a):\tp$, which we formalize by $\tp=\tp$.
And it maps typed terms $t:\tm a$ to untyped terms $\TE(t):\term$, which we formalize by $\tm=\tlam[\tp]{a}\term$ and thus $\TE(\tm\,a)=\term$.
We also use $\iincl{\Proofs}$ to include the identity morphism on \Proofs, i.e., all constants of \Proofs are mapped to themselves.

Applying this morphism, i.e., the pushout functor $\PO{\TE}$, to the theory \HProd of hard-typed simple products yields the theory $\HProd^\TE$, which arises by replacing every occurrence of $\tm A$ with $\term$.
\TE also extends to the morphism $\TE_\HProd$, which translates all expressions of \HProd to expressions of $\HProd^\TE$.
This translations preserves LF-typing, e.g., if $\oftype{\HTyped}{}{t}{\tm \prod\,a\,b}$, then $\oftype{\HTyped^\TE}{}{\TE_\HProd(t)}{\term}$.

However, $\HProd^\TE$ is not the desired formalization of soft-typed products (e.g., because it lacks constants relating types and terms), and we develop a more suitable functor in the next section.
\end{example}

\section{The Softening Operator}\label{sec:logrel}

\subsection{Basic Overview}\label{sec:basicsoften}

\Soften translates theories based on \HTyped to theories based on \STyped.
The key idea is that whenever we have an expression $t : \tm a$ in $\HTyped$, then in $\STyped$ we need to synthesize two things: an expression $\TE(t) : \term$ and an expression $t^* : \of \TE(t)\,\TE(a)$ acting as a witness of type preservation.
And whenever we have an expression $a : \tp$, we need to synthesize one thing only, namely $\TE(a) : \tp$.
(Note that for simple types such as product and function types, we have $\TE(a) = a$. We discuss dependent function types in the next section.)
Both intuitions extend homomorphically to all concepts of LF such as function types and contexts.

As an example, consider the constants $\pair,\projL,\projR$ in $\HProd$ in Fig.~\ref{fig:prod}.
For each of them we synthesize a type-erased constant of the same name and a starred typing witness in $\STyped$.
Note that the type parameters $a$ and $b$ are removed in their corresponding type-erased constant in $\SProd$.
We have been unable to find a systematic way to determine when arguments need to be removed and discuss this problem in Sec.~\ref{sec:remove-params}.
For the arguments in the starred constants such as $\pair^*$, we synthesize two parameters $x : \term$ and $x^* : \ded \of x\,a$ (whose name can often be omitted).

\begin{figure}[hbt]
\begin{multicols}{2}
	\begin{mmtmods}
		\mthy{\HProd}{}\\
		\mincl{\HTyped}{}\\
		\mcons{\prod}{\tp\to\tp\to\tp}\\
		\mcons{\pair}{\tPi{a,b}\tm a\to\tm b\to\tm \prod\,a\,b}\\
		\mcons{\projL}{\tPi{a,b}\tm \prod\,a\,b\to\tm a}\\
		\mcons{\projR}{\tPi{a,b}\tm \prod\,a\,b\to\tm b}\\
		\mend\\{}\\{}\\{}\\
	\end{mmtmods}
	\begin{mmtmods}
		\mthy{\SProd}{}\\
		\mincl{\STyped}{}\\
		\mcons{\prod}{\tp\to\tp\to\tp}\\
		\mcons{\pair}{\term\to\term\to\term}\\
		\mcons{\pair^*}{\tPi{a,b}\tPi{x}\ded \of x\,a\to\tPi{y}\ded \of y\,b \mnl \to\ded\of(\pair\,x\,y)\,(\prod\,a\,b)}\\
		\mcons{\projL}{\term\to\term}\\
		\mcons{\projL^*}{\tPi{a,b}\tPi{x}\ded \of x\,(\prod\,a\,b)\mnl\to\ded\of(\projL\,x)\,a}\\
		\mcons{\projR}{\term\to\term}\\
		\mcons{\projR^*}{\tPi{a,b}\tPi{x}\ded \of x\,(\prod\,a\,b)\mnl\to\ded\of(\projR\,x)\,b}\\
		\mend
	\end{mmtmods}
\end{multicols}
\caption{Hard and Soft Product Types}\label{fig:prod}
\end{figure}

\subsection{Logical Relations}\label{sec:soften:logrel}

We capture the type preservation proof using logical relations.
The meta-theory for using logical relations to represent type preservation was already sketched in~\cite{RS:logrels:12}, but we have to make a substantial generalization to \emph{partial} logical relations and extend those to natural functors.
Besides allowing the representation of partial translations, this has an important practical advantage: the translations in~\cite{RS:logrels:12} must introduce unit argument types in places where no particular property about a term is proved.
While semantically irrelevant, softening must remove these in order to produce the softened theories actually expected by humans, thus potentially violating the correctness of the translation.
Partiality allows constructing the softened theories in a way that these artefacts are not introduced in the first place.

Because logical relations can be very difficult to wrap one's head around, we focus on the special case needed for softening although we have designed and implemented it for the much more general setting of~\cite{RS:logrels:12}.
Moreover, we advise readers to maintain the following intuitions while perusing the formal treatment below:
\begin{compactitem}
	\item The morphism $m:S\to T$ is the type erasure translation $\TE:\HTyped\to\STyped$.
	\item \label{def:tp}The logical relation $r$ is a mapping $\TP$ from \HTyped-syntax to \STyped-syntax that maps
	\begin{compactitem}
		\item types $A:\type$ to unary predicates $\TP(A):\TE(A)\to\type$ about \TE-translated terms of type $A$
		\item terms $t:A:\type$ to proofs $\TP(t):\TP(A)\,\TE(t)$ of the predicate associated with $A$
	\end{compactitem}
	\item Even more concretely,
	\begin{compactitem}
		\item $\TP(\prop),\TP(\ded),\TP(\tp)$ are all undefined because we need not prove anything about terms at those types
		\item $\TP(\tmN)=\tlam[\tp]{a}\tlam[\term]{x}\of x\, a$ and thus $\TP(\tm a)=\tlam[\term]{x}\of x\, a$, i.e., $\TP$ maps every $t:\tm a$ to its typing proof $\TP(t):\of\,\TE(t)\,a$.
	\end{compactitem}
\end{compactitem}
Moreover, it may help readers to compare Def.~\ref{def:mor} and~\ref{def:rel} as well as Thm.~\ref{thm:mor} and~\ref{thm:rel}.

\begin{definition}\label{def:rel}
	\renewcommand{\r}{\ov{r}}
	A partial \textbf{logical relation} on a morphism $m:S\to T$ is a partial mapping $r$ of $S$-constants to $T$-expressions such that
	for every $S$-constant $c:A$, if $r(c)$ is defined, then so is $\r(A)$ and $\oftype{T}{}{r(c)}{\r(A)\,m(c)}$.
	$r$ is called \textbf{term-total} if it is defined for a typed constant if it is for the type.
	The partial mapping $\r$ of $S$-syntax to $T$-syntax is defined in Fig.~\ref{fig:rel}.
	In the sequel, we write $r$ for $\r$.
	\begin{figure}
		\[\eqns{
			\r(c) & r(c)\\[.2cm]
			\r(x) & \cas{x^* \mifc \text{$x^*$ was declared when traversing into the binder of $x$} \\ \text{undefined} \mothw}  \\
			\r(\type) & \tlam[\type]{a}a\to\type \\
			\r(\tPi{x:A}B) & \tlam[m(\tPi{x:A}B)]{f}\tPi{\r(x:A)}\r(B)\,(f\,x)\\
			\r(\tlam{x:A}t) & \tlam{\r(x:A)}\r(t)\\
			\r(f\,t) & \cas{\r(f)\,m(t)\,\r(t) \mifc \text{$\r(t)$ defined} \\ \r(f)\,m(t) \mothw}\\[.2cm]
			\r(\cdot) & \cdot\\
			\r(\Gamma,x:A) & \r(\Gamma), \cas{x:m(A),\,x^*:\r(A)\,x \mifc \text{$\r(A)$ defined} \\
				x:m(A) \mothw} \\
		}\]
		
		$\ov{r}(-)$ is undefined whenever an expression on the right-hand side is.
		\caption{Map induced by a Logical Relation}\label{fig:rel}
	\end{figure}
\end{definition}


The key idea of the map $\ov{r}$ is to attempt to construct $\ov{r}$ in the same way as in \cite{RS:logrels:12} for total $r$.
Whenever $\ov{r}$ is applied to an argument for which it is not defined, the expression is simply removed: if the type of a bound variable would be undefined, the whole binding is removed; if an argument of a function application would be undefined, the function is applied to one fewer argument.
The next theorem states that these removals fit together in the sense that $\ov{r}$ still satisfies the main property of logical relations whenever it is defined:

\begin{theorem}\label{thm:rel}
	For a partial logical relation $r$ on a morphism $m:S\to T$, we have
	\begin{compactitem}
		\item if $\oftype{S}{\Gamma}{t}{A}$ and $r$ is defined for $t$, then $r$ is defined for $A$ and $\oftype{T}{r(\Gamma)}{r(t)}{r(A)\,m(t)}$
		\item if $r$ is term-total, it is defined for a typed term if it is for its type
	\end{compactitem}
\end{theorem}
\begin{proof}
The inductive definition is the same as in \cite{RS:logrels:12} except for the possibility of undefinedness.
Thus, whenever the results are defined, the typing properties follow from the theorems there.

First, it is straightforward to see that $r$ is total on contexts and substitutions because the case distinctions explicitly avoid recursing into arguments for which $r$ is undefined.

Second, we show by induction on derivations of $\oftype{S}{\Gamma}{t}{A}$ that if $A:\type$ then $r$ is defined for $t$ iff it is defined for $A$.
\begin{compactitem}
\item constant $c:A$: True by assumption.
\item variable $x:A$: The case for $\Gamma,x:A$ introduces the variable $x^*$ into the target context if $r(A)$ is defined.
 The case for $x$ picks up on that and (un)defines $r$ at $x$ accordingly.
\item $\lambda$-abstraction $\tlam{x:A}t:\tPi{x:A}B$: $r$ is always defined for $x:A$.
 By induction hypothesis, it is defined for $t$ if it is for $B$.
\item $t$ cannot be a $\Pi$-abstraction
\item application $f\,t:B(t)$ for some $f:\tPi{x:A}B(x)$:
  By definition, $r$ is defined for $f\,t$ if it is defined for $f$.
  By induction hypothesis the latter holds iff $r$ is defined for $\tPi{x:A}B(x)$, which by definition holds iff it is defined for $B(x)$.
  It remains to show that $r$ is defined for $B(t)$ iff it is defined for $B(x)$ in the context extended with $x:A$.
  By induction hypothesis, $r$ is defined for $t$ iff it is defined for $x$.
  Therefore, and because the definition of $r$ is compositional, substituting $t$ for $x$ cannot affect whether $r$ is defined for an expression.
\end{compactitem}

Finally, if $\oftype{E}{\Gamma}{t}{A}$ for $A:\kind$, we need to show that $r$ is defined for $A$ if it is for $t$.
That is trivial: inspecting the definition shows that $r$ is always defined for kinds anyway.
\end{proof}

\medskip

We now capture how to synthesize $\SProd$ from $\HProd$ via logical relations.
First we define the type erasure morphism $\TE: \HTyped \to \STyped$ and the type preservation property as a logical relation $\TP$ on $\TE$.
Then we define and apply the following functor on $\HProd$, which can be thought of as an analog of pushout along a logical relation:
\begin{definition}\label{def:lrop}
	Consider a morphism $m:S\to T$ and a term-total logical relation $r$ on $m$.
	Then the functor $\LR{m}{r}$ from $S$-extensions to $T$-extensions maps theories $E$ as follows:
	\begin{compactenum}
		\item We compute the pushout $E^m:=\PO{m}(E)$.
		\item $E^m$ has the same shape as $E$ and there is a morphism $m_E:E\to E^m$.
		For each, we create an initially empty logical relation $r_E$ on $m_E$.
		\item For each declaration $c:A[=t]$ in $E$ for which $r_E(A)$ is defined, we add
		\begin{compactenum}
			\item the constant declaration $c^*:r_E(A)\,m_E(c)[=r_E(t)]$ to $E^m$ 
			\item the case $r(c)=c^*$ to $r_E$.
		\end{compactenum}
	\end{compactenum}
\end{definition}

Concretely we get:

\begin{mmtmods}
	\mthy{\LR{\TE}{\TP}(\HProd)}{}\\
	\mincl{\STyped}{}\\
	\mcons{\prod}{\tp\to\tp\to\tp}\\
	\mcons{\pair}{\tPi{a,b}\term\to\term\to\term}\\
	\mcons{\pair^*}{\tPi{a,b}\tPi{x}\ded \of x\,a\to\tPi{y}\ded \of y\,b \mnl \to\ded\of(\pair\,a\,b\,x\,y)\,(\prod\,a\,b)}\\
	\mcons{\projL}{\tPi{a,b}\term\to\term}\\
	\mcons{\projL^*}{\tPi{a,b}\tPi{x}\ded \of x\,(\prod\,a\,b) \to\ded\of(\projL\,a\,b\,x)\,a}\\
	\mcons{\projR}{\tPi{a,b}\term\to\term}\\
	\mcons{\projR^*}{\tPi{a,b}\tPi{x}\ded \of x\,(\prod\,a\,b) \to\ded\of(\projR\,a\,b\,x)\,b}\\
	\mend
\end{mmtmods}

Here we see that $\pair$, $\projL$, $\projR$ all take undesired (and unused) type arguments.
In the sequel, we will suitably extend the naive definition given above.

Fig.~\ref{fig:input-examples} gives some additional examples of hard-typed features.
Here we also include hard-typed equality \HEqual to formulate the reduction rules for function types.
Fig.~\ref{fig:output-examples} shows the corresponding soft-typed variants that we intend to obtain.
Note that these examples already foreshadow that \Soften can be extended to theories containing includes in a straightforward way.
We will define that formally in Sect.~\ref{sec:diagrams}.

\begin{figure}[hbt]
{
	\setlength{\columnsep}{-1em}
	\begin{multicols}{2}
	\setlength{\parindent}{0pt}
	\begin{mmtmods*}
		\mthy{\HEqual}{}\\
		\mincl{\HTyped}{}\\
		\keepParamAnnotation{\eq}{1}\\
		\mcons{\eq}{\tPi{a}\tm a\to\tm a\to\prop}\\
		\mcons{\refl}{\tPi{a,x}\ded \eq\,a\,x\,x}\\
		\mcons{\eqsub}{\tPi{a,x,y}\ded\eq\, a\,x\,y\to\mnl \tPi[\tm a\to\prop]{F}\ded F\,x\to\ded F\,y}\\
		\mend
	\end{mmtmods*}
	\begin{mmtmods*}
		\mthy{\HSimpFun}{}\\
		\mincl{\HEqual}{}\\
		\mcons{\fun}{\tp\to\tp\to\tp}\\
		\keepParamAnnotation{\lam}{1}\\
		\mcons{\lam}{\tPi{a,b}(\tm a \to \tm b)\to\tm \fun\, a\, b}\\
		\mcons{\app}{\tPi{a,b}\tm \fun\,a\,b \to \tm a\to \tm b}\\
		\mend
	\end{mmtmods*}	
	\begin{mmtmods*}
		\mthy{\HDepFun}{}\\
		\mincl{\HEqual}{}\\
		\keepParamAnnotation{\dfun}{1}\\
		\mcons{\dfun}{\tPi{a}(\tm a\to\tp)\to\tp}\\
		\keepParamAnnotation{\dlam}{1}\\
		\mcons{\dlam}{\tPi{a}\tPi[\tm a\to\tp]{b}(\tPi[\tm a]{x}\tm b\,x)\mnl\to\tm \dfun\, a\, b}\\
		\mcons{\dapp}{\tPi{a,b}\tm \dfun\,a\,b \to \tPi[\tm a]{x} \tm b\,x}\\
		\mend
	\end{mmtmods*}
	\end{multicols}
	\noindent
	\begin{multicols}{2}
	\setlength{\parindent}{0pt}
	\begin{mmtmods*}
		\mthy{\HBeta}{}\\
		\mincl{\HSimpFun}{}\\
		\mcons{\reduce}{%
			\tPi{a,b}\tPi[\tm a\to\tm b]{F}\tPi{x}\mnl\ded\eq\,b\,(\app\,a\,b\,(\lam\,a\,b\,F)\,x)\,(F\,x)%
		}\\
		\mend
	\end{mmtmods*}
	\begin{mmtmods*}
		\mthy{\HEta}{}\\
		\mincl{\HSimpFun}{}\\
		\mcons{\repr}{\tPi{a,b}\tPi[\tm \fun\,a\,b]{f}\mnl \ded \eq\,(\fun\,a\,b)\,f\,(\lam\,a\,b\,\tlam{x}\app\,f\,x)}\\
		\mend
	\end{mmtmods*}
	\begin{mmtmods*}
		\mthy{\HExten}{}\\
		\mincl{\HSimpFun}{}\\
		\mcons{\exten}{%
			\tPi{a,b}\tPi[\tm \fun\,a\,b]{f,g}\mnl%
			(\tPi{x}\ded \eq\,b\,(\app\,a\,b\,f\,x)\mnl%
			(\app\,a\,b\,g\,x)) \to \ded \eq\,(\fun\,a\,b)\,f\,g%
		}\\
		\mend
	\end{mmtmods*}
	\begin{mmtmods*}
		\mthy{\HDepBeta}{}\\
		\mincl{\HDepFun}{}\\
		\mcons{\dreduce}{\tPi{a,b}\tPi[{\tPi[\tm a]{x}\tm b\,x}]{F}\tPi{x} \mnl \ded \eq\,(b\,x)\,(\dapp\,a\,b\,(\dlam\,a\,b\,F)\, x)\,(F\,x)}\\
		\mend
	\end{mmtmods*}
	\end{multicols}
}
	\caption{\label{fig:input-examples}Theories for Function Types with Annotations for Needed Arguments}
\end{figure}

\begin{figure}[hbt]
{
	\setlength{\columnsep}{-4.1em}
	\begin{multicols}{2}
		\setlength{\parindent}{0pt}
		\begin{mmtmods*}
			\mthy{\SEqual}{}\\
			\mincl{\STyped}{}\\
			\mcons{\eq}{\tPi{a}\term\to\term\to\prop}\\
			\mcons{\refl^*}{\tPi{a,x}\ded \of x\,a \to \ded \eq\,a\,x\,x}\\
			\mcons{\eqsub^*}{%
				\tPi{a}\tPi[\ded \of x\,a]{x,x^*}\tPi[\ded \of y\,a]{y,y^*}\mnl%
				\ded \eq\,a\,x\,y \to\mnl%
				\tPi[\term \to \prop]{F} \ded F\,x \to%
				\ded F\,y%
			}\\
			\mend
		\end{mmtmods*}
		\begin{mmtmods*}
			\mthy{\SSimpFun}{}\\
			\mincl{\SEqual}{}\\
			\mcons{\fun}{\tp\to\tp\to\tp}\\
			\mcons{\lam}{\tPi{a} (\term \to \term) \to \term}\\
			\mcons{\lam^*}{%
				\tPi{a,b}\tPi[\term \to \term]{F}\mnl%
				(\tPi{x} \ded \of x\,a \to \ded \of (F\,x)\,b)\mnl%
				\to\ded \of (\lam\,a\,F)\,(\fun\,a\,b)%
			}\\
			\mcons{\app}{\term\to\term\to\term}\\
			\mcons{\app^*}{%
				\tPi{a,b} \tPi{f} \ded \of f\,(\fun\,a\,b)\to\mnl%
				\tPi{x} \ded \of x\,a \to \ded \of (\app\,f\,x)\,b%
			}\\
			\mend
		\end{mmtmods*}
		\begin{mmtmods*}
			\mthy{\SDepFun}{}\\
			\mincl{\SEqual}{}\\
			\mcons{\dfun}{\tPi{a}(\term \to \tp)\to\tp}\\
			\mcons{\dlam}{\tPi{a} (\term \to \term) \to \term}\\
			\mcons{\dlam^*}{%
				\tPi{a}\tPi[\term\to\tp]{b}\tPi[\term\to\term]{F}\mnl%
				(\tPi{x} \ded \of x\,a \to \ded \of (F\,x)\,(b\,x))\mnl%
				\to \ded \of (\dlam\,a\,b\,F)\,(\dfun\,a\,b)%
			}\\
			\mcons{\dapp}{\term \to \term \to \term}\\
			\mcons{\dapp^*}{%
				\tPi{a,b} \tPi{f} \ded \of f\,(\dfun\,a\,b)\to\mnl%
				\tPi{x} \ded \of x\,a \to \ded \of (\dapp\,f\,x)\,(b\,x)%
			}\\
			\mend
		\end{mmtmods*}
	\end{multicols}
	\begin{multicols*}{2}
		\setlength{\parindent}{0pt}
		\begin{mmtmods*}
			{}\\
			\mthy{\SBeta}{}\\
			\mincl{\SSimpFun}{}\\
			\mcons{\reduce^*}{
				\tPi{a,b}\tPi[\term\to\term]{F}\mnl%
				(\tPi{x} \ded \of x\,a \to \ded \of (F\,x)\,b)\mnl%
				\to\tPi{x}\ded \of x\,a \to\mnl%
				\ded\eq\,b\, (\app\,(\lam\,a\,F)\, x)\,(F\,x)%
			}\\
			\mend
		\end{mmtmods*}
		\begin{mmtmods*}
			\mthy{\SEta}{}\\
			\mincl{\SSimpFun}{}\\
			\mcons{\repr^*}{%
				\tPi{a,b}\tPi[\term]{f} \ded \of f\,(\fun\,a\,b)\mnl%
				\ded \eq\,(\fun\,a\,b)\,f\,(\lam\,a\,\tlam{x}\app\,f\,x)%
			}\\
			\mend
		\end{mmtmods*}
		\begin{mmtmods*}
			{}\\
			\mthy{\SExten}{}\\
			\mincl{\SSimpFun}{}\\
			\mcons{\exten^*}{%
				\tPi{a,b}\tPi[\term]{f} \ded \of f\,(\fun\,a\,b) \to\mnl%
				\tPi[\term]{g} \ded \of g\,(\fun\,a\,b) \to\mnl%
				(\tPi{x} \ded \of x\,a \to \ded \eq\,b\,(\app\,f\,x)\,(\app\,g\,x))\mnl%
				\to \ded \eq\,(\fun\,a\,b)\,f\,g%
			}\\
			\mend
		\end{mmtmods*}
		\begin{mmtmods*}
			\mthy{\SDepBeta}{}\\
			\mincl{\SDepFun}{}\\
			\mcons{\dreduce^*}{%
				\tPi{a,b}\tPi[\term\to\term]{F}\mnl%
				(\tPi{x} \ded \of x\,a \to \ded \of (F\,x)\,(b\,x))\to\mnl%
				\tPi{x} \ded \of x\,a \to\mnl%
				\ded \eq\,(b\,x)\,(\dapp\,(\dlam\,a\,F)\,x)\,(F\,x)}\\
			\mend
		\end{mmtmods*}
	\end{multicols*}
}
	\caption{\label{fig:output-examples}Result of Softening the Theories from Fig.~\ref{fig:input-examples}}
\end{figure}

\subsection{Removal of Unnecessary Parameters}\label{sec:remove-params}

In Sect.~\ref{sec:soften:logrel} we developed a translation from $\HTyped$ to $\STyped$ that maps every constant $c: A$ to a translated constant $c: m(A)$ and a witness $c^*: r(A)\,c$, where we chose $m = \TE$ to be our type erasure morphism and $r = \TP$ our logical relation capturing type preservation.
This translation almost produced the desired formalization $\SProd$ except that some translated constants featured undesired type parameters.
Pre- or post-composing our translation with one that removes selected type parameters is straightforward and presented in the following.
The major problem is identifying these parameters in the first place.
For example, in the library in Fig.~\ref{fig:input-examples} above we can distinguish the following cases:
\begin{compactitem}
	\item removal required, e.g., $\pair: \tPi{a,b}\tm a\to\tm b\to\tm \prod\,a\,b$ should go to
		$\pair: \term\to\term\to\term$
	\item removal optional depending on the intended result, e.g., $\eq: \tPi{a}\tm a\to\tm a\to\prop$ can go to $\eq: \term\to\term\to\prop$ or to $\eq: \tPi{a}\term\to\term\to\prop$;
	     analogously $\lam: \tPi{a,b} (\tm a \to \tm b) \to \tm \fun\,a\,b$ can go to $\lam: (\term \to \term) \to \term$ or to $\lam: \tPi{a} (\term \to \term) \to \term$
	\item removal forbidden, e.g., $\dfun: \tPi[\tp]{a}(\tm a\to\tp)\to\tp$ should go to
		$\dfun: \tPi{a}(\term\to\tp)\to\tp$
\end{compactitem}

\begin{definition}[Unused Positions]
	Consider a constant $c:A$ in a theory $S$ in a diagram $D$.
	After suitably normalizing, $A$ must start with a (possibly empty) sequence of $n$ $\Pi$-bindings, and any definition of $c$ (direct or morphism) must start with the same variable sequence $\lambda$-bound.
	We write $c^1,\ldots,c^n$ for these variable bindings.
	Each occurrence of $c$ in an expression in $D$ is (after suitably $\eta$-expanding if needed) applied to exactly $n$ terms, and we also write $c^i$ for those argument positions.
	
	We call a set $P$ of argument positions of $D$-constants \textbf{unused} if for every $c^i\in P$, the $i$-th bound variable of the type or any definition of $c$ occurs at most as a subexpression of argument positions that are themselves in $P$.
	
	We write $D\sm P$ for the diagram that arises from $P$ by removing for every $c^i\in P$
	\begin{compactitem}
		\item the $i$-th variable binding in the type and all definitions of $c$, e.g., $c:\tPi[A_1]{x_1}\tPi[A_2]{x_2}B$ becomes $c:\tPi[A_1]{x_1}B$ if $i=2$,
		\item the $i$-argument of any application of $c$, e.g., $c\,t_1\,t_2$ becomes $c\,t_1$ if $i=2$.
	\end{compactitem}
\end{definition}

\begin{lemma}[Removing Unused Positions]\label{thm:remove}
	Consider a well-typed diagram $D$ and a set $P$ of argument positions unused in $D$.
	Then $D\sm P$ is also well-typed.
\end{lemma}
\begin{proof}
Technically, this is proved by induction on the typing derivation of $D$.
But it is easy to see: by construction, (i) the variables bindings in $P$ do not occur in $D\sm P$ so that all types and definitions stay well-typed, and (ii) the type, definitions, and uses of all constants are changed consistently so that they stay well-typed.
The only subtlety is that we need to apply LF's $\eta$-equality to expand not fully applied uses of a constant.
\end{proof}

Note that, in the presence of include declarations or morphisms, the decision whether an argument position may be removed is not local: we must consider the entire diagram to check for all occurrences.
If a theory $T$ includes the theory $S$ and uses a constant $c$ declared in $S$, then an argument position $c^i$ may be unused in $S$ but used in $T$.
Thus, the functor that removes argument positions may have to be undefined on $T$.

Implementing the operation $D\sm P$ is straightforward.
However, much to our surprise and frustration, automatically choosing an appropriate set $P$ turned out to be difficult:
\begin{example}\label{ex:remove}
	The undesired argument positions in $\TE^\HProd$ are exactly the named variables in \HProd that do not occur in their scopes in $\TE^\HProd$ anymore. This includes the positions $\pair^1$ and $\pair^2$, and removing them yields the desired declaration of \pair in \SProd.
	
	However, that does not hold for \HDepFun.
	Here the argument $\dfun^1$ is named in \HDepFun and unused in the declaration $\dfun:\tPi[\tp]{a}(\term\to\tp)\to\tp$ that occurs in $\TE^\HDepFun$.
	However, that is in fact the desired formalization of the soft-typed dependent function type.
	Removing $\dfun^1$ would yield the undesired $\dfun:(\term\to\tp)\to\tp$.
	While we do not mention \mmt's implicit arguments in this paper, note also that $\dfun^1$ is an \emph{implicit} argument in \HDepFun that must become \emph{explicit} in \SDepFun.
\end{example}

This is trickier than it sounds because some argument positions may only be removable if they are removed at the same time; so a fixpoint iteration might be necessary.
Moreover, picking the largest possible $P$ is entirely wrong as it would remove all argument positions.
At the very least, we should only remove \emph{named} argument positions, i.e., those that are bound by a named variable (as opposed to the anonymous variables introduced by parsing e.g., $\prod:\tp\to\tp\to\tp$).
A smarter choice is to remove all named argument positions that become redundant through pushout, e.g., that are named and used in $\HProd$ but unused in $\TE(\HProd)$.
(Note that the pushout $\PO{m}(D)$ has at least the argument positions that $D$ has. It may have more if $m$ maps an atomic type to a function type.)
That is the right choice almost all the time but not always.


After several failed attempts, we have been unable to find a good heuristic for choosing $P$.
For now, we remove all named variables that never occur in their scope anymore, and we allow users to annotate constants like $\keepParam{\dfun}{1}$ where the system should deviate from that heuristic (see Fig.~\ref{fig:input-examples}).
We anticipate finding better solutions after collecting more data in the future.
In the sequel, we write $\POc{m}(D):=\PO{m}(D)\sm P_D$ where $P_D$ is any fixed heuristic.
$\HProd^{\POc{\TE}}$ yields the theory \SProd except that it still lacks the $^*$-ed constants.
The following lemma shows that we can now obtain the morphism $e:\HProd\to\SProd$ from above as $\POc{\TE}_\HProd$:

\begin{lemma}[Removing Arguments Preserves Naturality]\label{thm:remnat}
	Consider a natural functor $O$ and a functor $O'(D):=O(D)\sm P_D$ for some heuristic $P$.
	Then $O'$ is natural as well.
\end{lemma}
\begin{proof}
$O$ being natural yields morphisms $O_E:E\to E^O$ from $D$-theories to $O(D)$ theories.
$O(D')$ has the same shape as $O(D)$, and to show that $O'(D)$ is natural, we reuse essentially the same morphisms from $D$-theories to $O'(D)$-theories.
We only have to $\eta$-expand the right-hand sides of all assignments in the morphisms $O_E$ and remove the same argument positions in $P_D$ as well.
\end{proof}

It is straightforward to extend Def.~\ref{def:rel} to all theories extending $S$ in the same way as pushout extends a morphism.
That would yield an include- and definition-preserving natural functor.
However, we omit that here because that functor would work with $\PO{m}$ whereas we want to use $\POc{m}$.
Instead, we make a small adjustment similar how we obtained $\POc{m}$ from $\PO{m}$:

\begin{definition}\label{def:lrop:cleaned}
	Consider a morphism $m:S\to T$ and a term-total logical relation $r$ on $m$.
	Then the functor $\LR{m}{r}$ maps a theory $E$ as follows:
	\begin{compactenum}
		\item We compute $E^m=\POc{m}(E)$.
		\item Due to Lem.~\ref{thm:remnat}, $E^m$ has the same shape as $E$, and there is a morphism $m_E:E\to E^m$.
		We create an initially empty logical relation $r_E$ on $m_E$.
		\item For each $E$-declaration $c:A[=t]$ for which $r_E(A)$ is defined, we add
		\begin{compactenum}
			\item the constant declaration $c^*:r_E(A)\,c[=r_E(t)]$ to $E^m$ 
			\item the case $r(c)=c^*$ to $r_E$.
		\end{compactenum}
	\end{compactenum}
\end{definition}

\begin{theorem}\label{thm:lrop}
	In the situation of Def.~\ref{def:lrop:cleaned}, the operator $\LR{m}{r}$ is a natural functor. 
	And every $r_E$ is a term-total logical relation on $m_E$.
\end{theorem}
\begin{proof}
We already know that $\POc{-}$ has the desired properties.
Moreover, adding well-typed declarations to $\POc{m}$ does not affect the naturality (because adding declaration to the codomain never affects the well-typedness of a morphism).
So for the first claim, we only have to prove that our additions are well-typed.

We prove that and the fact that $r_E$ is a logical relation jointly by induction on the derivation of the well-typedness of $D$:
We appeal to Thm.~\ref{thm:rel} to show that the added constant declarations are well-typed.
And the cases $r(c)=c^*$ satisfy the typing requirements of logical relations by construction.
\end{proof}

Now the functor $\LR{\TE}{\TP}$ generates for every hard-typed feature $F$
\begin{compactitem}
	\item the corresponding soft-typed feature $F'$
	\item the type-erasure translation $\TE_F:F\to F'$ as a compositional/homomorphic mapping,
	\item the type preservation proof $\TP_F$ for the type erasure as a logical relation on $\TE_F$.
\end{compactitem}
In particular, we have $\SProd=\LR{\TE}{\TP}(\HProd)$.

\subsection{Translating Proof Rules Correctly}\label{sec:proof-rules}

We omitted the reduction rules in our introductory example \HProd.
This was because $\LR{\TE}{\TP}$ is still not the right operator.
To see what goes wrong, assume we leave $\TP(\dedN)$ undefined, and consider the type of the \reduce rule from \HBeta:

\begin{center}
	\begin{tabular}{ll}
		\HBeta &  $\tPi{a,b}\tPi[\tm a\to\tm b]{F}\tPi{x}\ded\eq\,b\, (\app\,a\,b\,(\lam\,a\,b\,F)\, x)\,(F\,x)$\\
		$\HBeta^{\LR{\TE}{\TP}}$ (generated) & $\tPi{a,b}\tPi[\term\to\term]{F}\tPi{x}$ \\ & \tb $\ded\eq\,b\, (\app\,(\lam\,a\,F)\, x)\,(F\,x)$ \\
		\SBeta (needed) & $\tPi{a,b}\tPi[\term\to\term]{F}\tPi[\tPi{a}\ded\of x\,a\to\ded\of (F\,x)\,b]{F^*}\tPi{x}\tPi[\ded\of x\,a]{x^*}$ \\ & \tb $\ded\eq\,b\, (\app\,(\lam\,a\,F)\, x)\,(F\,x)$
	\end{tabular}
\end{center}

The rule generated by $\LR{\TE}{\TP}(\HProd)$ is well-typed but not sound.
In general, the softening operator must insert $^*$-ed assumptions for all variables akin to how Def.~\ref{def:lrop:cleaned} inserts them for constants.
But it must only do so for proof rules and not for, e.g., \fun, \lam, and \app.

We can achieve that by generalizing to partial logical relations on \emph{partial} morphisms.
Intuitively, we define $\LRc{m}{r}$ for partial $m$ and $r$ in the same way as $\LR{m}{r}$, again dropping all variable and constant declarations for whose type the translation is partial.

First we refine \TE and \TP (from Fig.~\ref{fig:te-pushout} and Page~\pageref{def:tp}, respectively) as follows:
\begin{compactitem}
	\item We leave $\TE(\dedN)$ undefined, i.e., our morphisms do not translate proofs. That is to be expected because we know that \TE cannot be extended to a morphism that also translates proofs \cite{rabe:lax:14}.
	\item We put $\TP(\dedN) = \tlam[\prop]{p}\ded p$ and thus $\TP(\ded P)=\ded \TE(P)$ for all $P$.
	This trick has the effect that $\reduce^*$ is generated as well and has the needed type (whereas the generation of \reduce is suppressed).
\end{compactitem}
Then we finally define $\Soften=\LRc{\TE}{\TP}$.
For every proof rule $c$ over \HTyped, it
\begin{compactitem}
	\item drops the declaration of $c$,
	\item generates the declaration of $c^*$, which now has the needed type.
\end{compactitem}

\Soften is still include- and definition-preserving but is no longer natural.
We conjecture that it is lax-natural and captures proof translations as lax morphisms in the sense of \cite{rabe:lax:14}.

%

\section{Translating Libraries}\label{sec:diagrams}
In the examples so far we have applied \Soften on theories extending \HTyped one at a time.
We now extend it to a translation on whole structured diagrams of theories and morphisms, mapping whole libraries of hard-typed features at once.

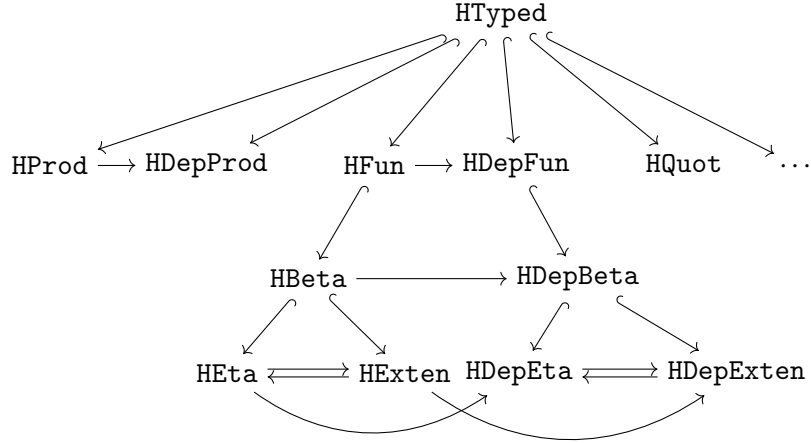
\begin{figure}
	\[\begin{tikzpicture}[right hook->,x=2cm,y=2cm]
		\node at (3,5) (HTyped)   {\HTyped};
		\node at (0,4) (HSimpProd)    {\HSimpProd};
		\node[right=0.5cm of HSimpProd] (HDepProd) {\HDepProd};
		\node[right=0.75cm of HDepProd] (HSimpFun) {\HSimpFun};
		\node[right=0.5cm of HSimpFun] (HDepFun)  {\HDepFun};
		\node[right=0.75cm of HDepFun] (HSubtyping) {\HQuot};
		\node[right=0.5cm of HSubtyping] (HQuotiening) {\ldots};
		\node[below left=1cm and -0.3cm of HSimpFun] (HBeta) {\HBeta};
		\node[below left =0.75cm and -0.1cm of HBeta] (HEta) {\HEta};
		\node[below right=0.75cm and -0.1cm of HBeta] (HExten) {\HExten};
		
		\node[right=2cm of HBeta] (HDepBeta) {\HDepBeta};
		\node[right=2.5cm of HEta] (HDepEta) {\HDepEta};
		\node[right=1cm of HDepEta] (HDepExten) {\HDepExten};
		
		\draw (HTyped) -- (HSimpProd);
		\draw (HTyped) -- (HDepProd);
		\draw (HTyped) -- (HSimpFun);
		\draw (HTyped) -- (HDepFun);
		\draw (HTyped) -- (HSubtyping);
		\draw (HTyped) -- (HQuotiening);
		\draw (HSimpFun) -- (HBeta);
		\draw (HBeta) -- (HEta);
		\draw (HBeta) -- (HExten);
		\draw (HDepFun) -- (HDepBeta);
		\draw (HDepBeta) -- (HDepEta);
		\draw (HDepBeta) -- (HDepExten);
		
		\draw[->] (HSimpProd) -- (HDepProd);
		\draw[->] (HSimpFun) -- (HDepFun);
		\draw[->,transform canvas={yshift=0.3ex}] (HEta) -- (HExten);
		\draw[->,transform canvas={yshift=-0.3ex}] (HExten) -- (HEta);
		
		\draw[->,transform canvas={yshift=0.3ex}] (HDepEta) -- (HDepExten);
		\draw[->,transform canvas={yshift=-0.3ex}] (HDepExten) -- (HDepEta);
		
		\draw[->] (HBeta) -- (HDepBeta);
		\draw[->,bend angle=35,bend right] (HEta) to (HDepEta);
		\draw[->,bend angle=35,bend right] (HExten) to (HDepExten);
	\end{tikzpicture}\]
	\caption{\label{fig:diagrams:overview}Diagram of hard-typed features}
\end{figure}

Before spelling out the definition, we show an exemplary library of hard-typed features in Fig.~\ref{fig:diagrams:overview}.
Here, we extend the collection of theories shown so far, the most notable extensions being several morphisms and the theory \HQuot formalizing hard-typed quotient types.
Here, the morphism $\HSimpProd \to \HDepProd$ realizes simple product types as a special case of dependent product types.
Analogously, all of the morphisms $\cn{H}\{\cn{Fun},\cn{Beta},\cn{Eta},\cn{Exten}\} \to \cn{HDep}\{\cn{Fun},\cn{Beta},\cn{Eta},\cn{Exten}\}$
realize the simply-typed feature as a special case of the corresponding dependently-typed feature.
And the anti-parallel morphism pairs $\HEta \xtofrom[]{} \HExten$ and $\HDepEta \xtofrom[]{} \HDepExten$ capture that $\eta$ and extensionality are equivalent in the presence of $\beta$ reduction.
Our goal is to soften this library in a way that preserves the modular structure.

\begin{definition}[\Soften on Diagrams]\label{def:soften:diagrams}
	On the category of theories and partial morphisms, we define \Soften as the partial functor translating diagrams $D$ over \HTyped to diagrams $D'$ over \STyped as follows:
	\begin{compactitem}
		\item every theory/morphism with name $X$ in $D$ yields a theory/morphism $X^\Soften$ in $D'$
		\item every $\incl{\HTyped}$ is replaced by $\incl{\STyped}$; and every $\incl{X}$ by $\incl{X^\Soften}$		
		\item every declaration $c: A\,[=t]$ in a theory $S$ yields those on the left below, and every assignment $c := t$ in a morphism yields those on the right below (whenever the involved translations are defined)
		\begin{align*}
			&c   &&\hspace*{-6.5em}: m_S(A)\,[=m_S(t)]\quad&&c&&\hspace*{-6.5em}:= m_S(t)\\
			&c^* &&\hspace*{-6.5em}: \TP^S(A)\,c\,[=\TP^S(t)]&&c^* &&\hspace*{-6.5em}:= \TP^S(t)
		\end{align*}
		where $m_S: S \to \POc{\TE^S}(S)$ is the morphism from Lem.~\ref{thm:remnat} and $\TE^S$ and $\TP^S$ are given below.
	\end{compactitem}
	
	We define $\TE$ as the partial morphism and $\TP$ as the partial logical relation on $\TE$ by
		\[\eqns{
			\TE(\prop) & \prop\\
			\TE(\tp)   & \tp\\
			\TE(\tmN)   & \tlam[\tp]{a} \term\\[.17cm]
			\TP(\dedN)    & \tlam[\prop]{p} \ded p\\
			\TP(\tmN)     & \tlam[\tp]{a} \tlam[\term]{x} \of\,x\,a
		}\]

   Then we additionally build the following components of $D'$:	
  \begin{compactitem}
	  \item every theory $X$ yields a partial morphism $\TE^T: T \to T^\Soften$ and a partial logical relation $\TP^T$ over $\TE^T$
    \item every $\incl{\HTyped}$ in a theory $S$ is replaced by $\incl{\TE}$ in $\TE^S$, and $\TP^S$ is made to extend $\TP$; and every other $\incl{T}$ in theories $S$ yields $\incl{\TE^S}$, and the definition of $\TP^S$ extends $\TP^T$
    \item every declaration $c: A\,[=t]$ in a theory $S$ yields $c := c$ in $\TE^S$ and $\TP^S(c) := c^*$ (whichever are defined)
 \end{compactitem}
\end{definition}

\begin{theorem}[Structure Preservation]
 Consider the category of LF theories and partial type- and sub\-sti\-tu\-tion-preserving expression translations as morphisms.
 Then \Soften is functorial and preserves the structure of includes and definitions.
 It is natural via the morphisms $\TE^S$ and the relation $\TP^S$.
\end{theorem}
\begin{proof}This holds by construction.\end{proof}



This finally yields the intended soft-typed formulation of $\SSimpFun:=\HSimpFun^\Soften$ and $\SDepFun:=\HDepFun^\Soften$.
As an example, we give the morphism $\HSFtoDF: \HSimpFun \to \HDepFun$ and its translation below.

{
\newcommand{\mmorDense}[3]{\multicolumn{3}{l}{\lfkw{morph}\;#1\;:#2\kern-0.08em \to \kern-0.08em #3=}}
\setlength{\columnsep}{-5.7em}
\begin{multicols}{2}
\noindent
\hspace{0.05em}
\begin{mmtmods*}
	\mmorDense{\HSFtoDF}{\HSimpFun}{\HDepFun}\\
	\mincl{\HEqual}{}\\
	\mconsd{\fun}{}{\tlam{a,b}\dfun\,a\,\tlam{x}b}\\
	\mconsd{\lam}{}{\tlam{a,b,F}\dlam\,a\,(\tlam{x}b)\,F}\\
	\mconsd{\app}{}{\tlam{a,b,f,x}\dapp\,a\,(\tlam{x}b)\,f\,x}\\
	\mend
\end{mmtmods*}
\begin{mmtmods*}
	\mmorDense{\HSFtoDF^\Soften}{\HSimpFun^\Soften}{\HDepFun^\Soften}\\
	\mincl{\HEqual^\Soften}{}\\
	\mconsd{\fun}{}{\tlam{a,b}\dfun\,a\,\tlam{x}b}\\
	\mconsd{\lam}{}{\tlam{a,F}\dlam\,a\,F}\\
	\mconsd{\lam^*}{}{\tlam{a,b,F,F^*}\dlam^*\,a\,(\tlam{x} b)\,F\,F^*}\\
	\mconsd{\app}{}{\tlam{f,x}\dapp\,f\,x}\\
	\mconsd{\app^*}{}{\tlam{a,b,f,f^*,x,x^*}\dapp\,a\,(\tlam{x}b)\,f\,f^*\,x\,x^*}\\
	\mend
\end{mmtmods*}
\end{multicols}
}

If we generalize the meta-theory of \cite{rabe:lax:14} to partial morphisms/relations and work in a variant of LF that adds product types, we could pair up $\TE^S$ and $\TP^S$ into a single expression translation that maps every term to the pair of its type erasure and its type preservation proof.


\section{Implementation}\label{sec:impl}
The formalizations developed in this paper including the act of softening are available online as part of the LATIN2 library.\footnote{\url{https://gl.mathhub.info/MMT/LATIN2/-/tree/devel/source/casestudies/2021-softening}}
Our implementation adds a component to \mmt that applies logical relation-based translations to entire diagrams of theories and morphisms.
Then softening arises as one special case of that construction.
While all translations are implemented in the underlying programming language of \mmt and thus part of the trusted code base, our general and systematic approach minimizes the amount of new code needed for any given instance such as softening.
That makes it much easier to review their correctness.
In any case, the generated diagrams can easily be double-checked by the original logical framework.

We are still experimenting with how to trigger these translations.
It is non-obvious if softening should be triggered by a new kind of declaration in the logical framework, a library-level script that lives outsides the logical framework, or a feature of the implementation that transparently builds the softened theory whenever the user refers to it.
As a prototype, we have chosen the first of these approaches.

In our implementation it proved advantageous to not have separate syntax for morphisms and logical relations.
Indeed, both are maps of names to expressions that extend to compositional translations of expressions to expressions, differing only in the inductive extension.
Instead, we found a way to represent every logical relation as a morphism, thus obviating the need to introduce additional syntax for relations.
The trick is to implement a special include-preserving functor that generates a theory $I$ that specifies exactly the typing requirements for the cases in a logical relation $r$, and then to represent $r$ as a morphism out of $I$.
We can even use this trick to represent multiple translations at once in a single morphism.
For example, in our implementation we jointly represent $\TE$ and $\TP$ from Def.~\ref{def:soften:diagrams}
as a single morphism \lstinline|TypePres|,
which in implementation-near syntax reads as follows.

\begin{lstlisting}
view TypePres : HTyped_comptrans -> STyped =
  prop/TE = prop
  tp/TE   = tp
  tm/TE   = [x] term

  tm/TP   = [A,t] $\vdash$ t $\altcolon$ A
  ded/TP  = ded
\end{lstlisting}
Here, \verb|HTyped_comptrans| refers to a suitable kind interface theory for combined translations on \verb|HTyped|.

Softening now emerges as the composition of multiple operators in our implementation,
which for the sake of conciseness were combined in one big operator in this paper.
Assume we wanted to soften a library of hard-typed features given as a diagram \verb|HLibrary| (e.g., the one from Fig.~\ref{fig:diagrams:overview}).
First, we compute the diagram \verb|HLibrary_comptrans|
of corresponding interface theories.
This diagram has the same shape as \verb|HLibrary| and is a tree rooted in \verb|HTyped_comptrans|.
Then, we compute the pushout of the resulting diagram along \verb|TypePres|.
The steps so far are effectively equivalent to applying the operator from Def.~\ref{def:lrop} accounted for with correct translation of proof rules.
Finally, it remains to apply the operator that drops unnecessary parameters (according to the heuristic outlined in Sec.~\ref{sec:remove-params}).
Below, we show how the last two steps look like in our.
\begin{lstlisting}
diagram SLibrary :=
  DROP_PARAMS STyped (PUSHOUT TypePres HLibrary_comptrans)
\end{lstlisting}

\section{Conclusion}\label{sec:conc}

We have given a translation of hard-typed (intrinsic) to soft-typed (extrinsic) formalizations of type theory.
Even though the existence of such translation is known, it had previously proved difficult to derive it from meta-theoretic principles in such a way that it can be studied and implemented easily.
Our key insight was that the associated type preservation proof can be cast as a logical relation, which allowed us to derive the translation from the requirement that the logical relation proof succeeds.

Our translation preserves modularity, which makes it suitable for translating modular libraries of formalizations of various type theories.
That enhances the quality and coverage of the library while reducing the maintenance effort.
Our implementation will serve as a key component in scaling up our modular formalizations of type theories in our LATIN2 library.

We expect our methodology of functors on diagrams of LF theories to extend to other important translations such as adding polymorphism or universes.

\bibliography{../../macros/bib/rabe,../../macros/bib/systems,../../macros/bib/historical,../../macros/bib/pub_rabe,../../macros/bib/institutions}


\end{document}